\documentclass{article}
\usepackage[utf8]{inputenc}

\usepackage[margin=1in]{geometry}

\usepackage[algo2e,linesnumbered,ruled,vlined]{algorithm2e}
\usepackage{subcaption, makecell}

\usepackage{amsmath,amsfonts,amsthm,amssymb,color}
\usepackage{fancybox}
\usepackage{float}
\usepackage{hyperref}
\usepackage{mathabx}

\usepackage{tikz}
\usepgflibrary{arrows}
\usetikzlibrary{arrows.meta,automata,decorations.pathreplacing}

\usepackage{thm-restate}

\usepackage{graphicx}
\usepackage{textcomp}
\usepackage{xcolor}

\usepackage[most]{tcolorbox}
\usepackage{empheq}

\newtcbox{\mymath}[1][]{%
    nobeforeafter, math upper, tcbox raise base,
    enhanced, colframe=blue!30!black,
    colback=blue!30, boxrule=1pt,
    #1}

\newcommand{\dc}{\text{Dual-DWC}\xspace}

\newcommand{\kcco}{{\sc $k$-CCO}\xspace}

\newcommand{\spara}[1]{\smallskip{\bf #1}}

\newcommand{\mycomment}[1]{}

\newcommand{\hide}[1]{}

\makeatletter
\newcommand*\bigcdot{\mathpalette\bigcdot@{.5}}
\newcommand*\bigcdot@[2]{\mathbin{\vcenter{\hbox{\scalebox{#2}{$\m@th#1\bullet$}}}}}
\makeatother

\newcommand{\NPhard}{$\mathbf{NP}$-hard}

\newcommand{\argmax}{\mathop{\rm argmax}}

\newtheorem{problem}     {Problem}

\newcommand{\squishlist}{
 \begin{list}{$\bullet$}
  {  \setlength{\itemsep}{0pt}
     \setlength{\parsep}{3pt}
     \setlength{\topsep}{3pt}
     \setlength{\partopsep}{0pt}
     \setlength{\leftmargin}{2em}
     \setlength{\labelwidth}{1.5em}
     \setlength{\labelsep}{0.5em}
} }
\newcommand{\squishlisttight}{
 \begin{list}{$\bullet$}
  { \setlength{\itemsep}{0pt}
    \setlength{\parsep}{0pt}
    \setlength{\topsep}{0pt}
    \setlength{\partopsep}{0pt}
    \setlength{\leftmargin}{2em}
    \setlength{\labelwidth}{1.5em}
    \setlength{\labelsep}{0.5em}
} }

\newcommand{\squishdesc}{
 \begin{list}{}
  {  \setlength{\itemsep}{0pt}
     \setlength{\parsep}{3pt}
     \setlength{\topsep}{3pt}
     \setlength{\partopsep}{0pt}
     \setlength{\leftmargin}{1em}
     \setlength{\labelwidth}{1.5em}
     \setlength{\labelsep}{0.5em}
} }

\newcommand{\squishend}{
  \end{list}
}

 \let\OLDthebibliography\thebibliography
\renewcommand\thebibliography[1]{
  \OLDthebibliography{#1}
  \setlength{\parskip}{0pt}
  \setlength{\itemsep}{0pt plus 0.3ex}
}

\newtheorem{theorem}{Theorem}[section]

\theoremstyle{definition}

\newenvironment{fminipage}
{\begin{Sbox}\begin{minipage}}
		{\end{minipage}\end{Sbox}\fbox{\TheSbox}}

\begin{document}
	
	\title{Dense and well-connected subgraph detection in dual networks}

	\author{Tianyi Chen\\
		Boston University\\
		\texttt{ctony@bu.edu}
		\and
		Francesco Bonchi\\
		ISI Foundation, Turin\\
		Eurecat, Barcelona, Spain \\
		\texttt{ francesco.bonchi@isi.it}
		\and
		David Garcia-Soriano\\
		ISI Foundation, Turin\\
		\texttt{d.garcia.soriano@isi.it}
		\and
		Atsushi Miyauchi\\
		University of Tokyo\\
		\texttt{miyauchi@mist.i.u-tokyo.ac.jp}
		\and
		Charalampos E. Tsourakakis\\
		Boston University \& ISI Foundation\\
		\texttt{tsourolampis@gmail.com}
	}
	\maketitle
	
    \begin{abstract} 
    Dense subgraph discovery is a fundamental problem in graph mining with a wide range of applications~\cite{gionis2015dense}.
Despite a large number of applications ranging from computational neuroscience to social network analysis,  that take as input a {\em dual} graph, namely a pair of graphs on the same set of nodes, dense subgraph discovery methods focus on a single graph input with few notable exceptions \cite{semertzidis2019finding,charikar2018finding,reinthal2016finding,jethava2015finding}.  In this work, we focus the following problem: 


\begin{quotation}
\noindent Given a pair of graphs $G,H$ on the same set of nodes $V$, how do we find a subset of nodes $S \subseteq
V$ that induces a well-connected subgraph in $G$ and a dense subgraph in $H$?   
\end{quotation}

\noindent Our formulation generalizes previous research on dual graphs~\cite{Wu+15,WuZLFJZ16,Cui2018}, by enabling the   {\em control} of the connectivity constraint on $G$. We propose a novel mathematical formulation based on $k$-edge connectivity, and prove that it is solvable exactly in polynomial time. We compare our method to state-of-the-art competitors; we find empirically that ranging the connectivity constraint enables the practitioner to obtain  insightful information  that is otherwise inaccessible. Finally, we show that our proposed mining tool can be used  to better understand how users interact on Twitter,  and  connectivity aspects of human brain networks  with and without Autism Spectrum Disorder (ASD).

    \end{abstract}

    \section{Introduction}
\label{sec:introduction}

Dense subgraph discovery (DSD) is a major research area of graph mining, whose goal is to develop methods that extract one or more dense clusters, possibly overlapping, from  a static or time-evolving network~\cite{gionis2015dense}. There exist numerous dense subgraph discovery formulations, with different computational complexity, and a wide range of DSD-based real-world applications, see the tutorial by Gionis and Tsourakakis~\cite{gionis2015dense}. However, an increasing number of real-world problems involve a pair of graphs on the same vertex set.  For example, in neuroscience, an important pair of brain networks  is given by the  \emph{structural connectivity} and \emph{functional connectivity} graphs, defined on the basis of correlation on the time series of activity of the different brain regions (co-activation) \cite{brain,lanciano2020explainable}. In bioinformatics, differential co-expression network analysis of gene expression data is used to analyze gene-to-gene coexpression  jointly across two different networks, especially in cancer research \cite{Liu18}.
In computational social science, we are interested in understanding better how users interact on social media such as Twitter, where
for instance,  users may retweet each other, and may favorite certain tweets of other users. These interactions naturally induce two
graph layers, the {\em retweet} layer and the {\em favorite layer}.    Qi et al.~\cite{qi2012community}  use   two Flickr graph topologies to mine online communities; one topology is naturally induced by online friendships between users, while the second one is created from log files by adding an edge between two users if there exists a photo that has been ``liked'' by both of them. In reality mining, one is interested in
understanding relationships defined by Bluetooth scans and phone calls respectively \cite{eagle2006reality,dong2012clustering}. Furthermore, dual graphs are a special case of multi-layer networks \cite{magnani2011ml}, when the number of layers is equal to 2. Despite the large amount of research on DSD, relatively few works focus on  dense subgraph discovery across more than one graphs, see \cite{Wu+15,WuZLFJZ16,semertzidis2019finding,charikar2018finding,kim2015community,reinthal2016finding,jethava2015finding}; this is not a coincidence, as there exist few optimization problems that have been studied on pairs of graphs, see for instance the work of Bhangale et al. on the problem of simultaneous Max-Cut~\cite{bhangale2020simultaneous,bhangale2018near,bhangale2015simultaneous}. It is worth mentioning that there exist several works on  bi-criteria optimization problems where the two graph topologies coincide, potentially with two different sets of edge weights, e.g., \cite{ravi1996constrained,tsourakakis2018risk}. On the more applied side, there exist algorithmic heuristics and machine learning algorithms that mine multi-layer graphs for a variety of problems, ranging from core decomposition to link prediction, e.g.,   \cite{dong2012clustering,jiang2009mining,sotiropoulos2019twittermancer,silva2010structural,lanciano2020explainable,liu2020corecube,tsourakakis2019novel}, but these do not provide theoretical guarantees on the output quality. 

The existence of important applications creates a practical need for novel graph mining tools that allow to jointly mine two graph topologies on the same set of nodes. Towards this direction,  we introduce in this work the following problem that generalizes the works of Wu et al. \cite{WuZLFJZ16,Wu+15} and Cui et al.~\cite{Cui2018}:

\begin{tcolorbox}
\begin{problem}
\label{prob:prob1}
\normalfont
Given two simple, unweighted, undirected graphs $G=(V,E_G)$ and $H=(V,E_H)$,  find a set of nodes $S\subseteq V$ such that  $G[S]$ is well-connected and  $H[S]$ is dense.
\end{problem}
\end{tcolorbox}

\noindent In prior related work~\cite{Cui2018,WuZLFJZ16,Wu+15} the goal is to ensure $G[S]$ is (simply) connected and $H[S]$ is dense. To the best of our knowledge, our work is the first  that enables the control of the connectivity constraint over dual graphs, and this constitutes a powerful feature of our method  in practice. Specifically, our contributions include the following:

\vspace{2mm}

$\bullet$ {\bf Problem formulation.} We  formalize Problem~\ref{prob:prob1} using for the connectivity and density constraints the well-established graph theoretic notions of $k$-edge-connectivity and minimum degree respectively. Specifically, we study algorithmically the following problem: 

\begin{tcolorbox}
\begin{problem}
\label{prob:maxmin}
\normalfont
Given two simple, unweighted, undirected graphs $G=(V,E_G)$ and $H=(V,E_H)$, and positive integer $k$,  find a set of nodes $S\subseteq V$ such that  $G[S]$ is $k$-edge connected and the minimum degree in $H[S]$ is maximized.
\end{problem}
\end{tcolorbox}

 \begin{figure}
 \begin{tabular}{ccc}
   \includegraphics[width=0.3\textwidth]{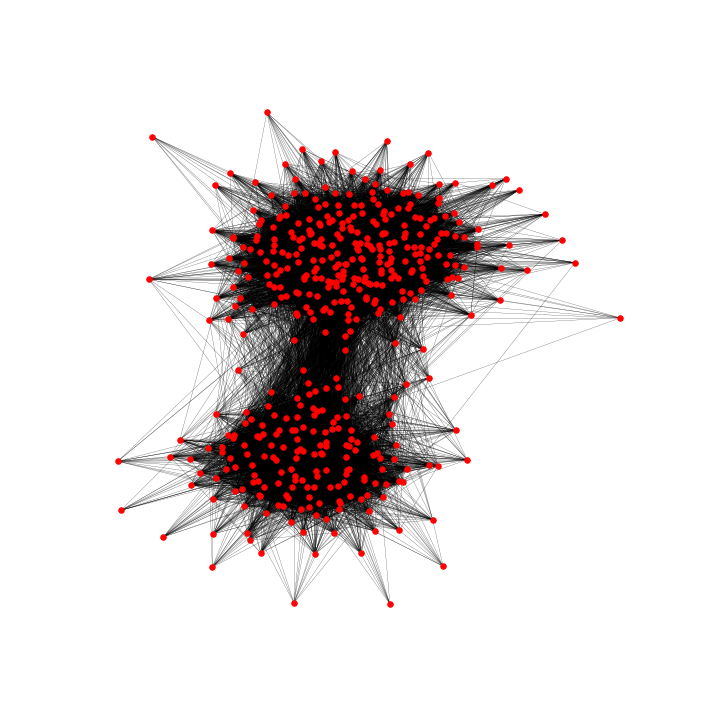} &    \includegraphics[width=0.3\textwidth]{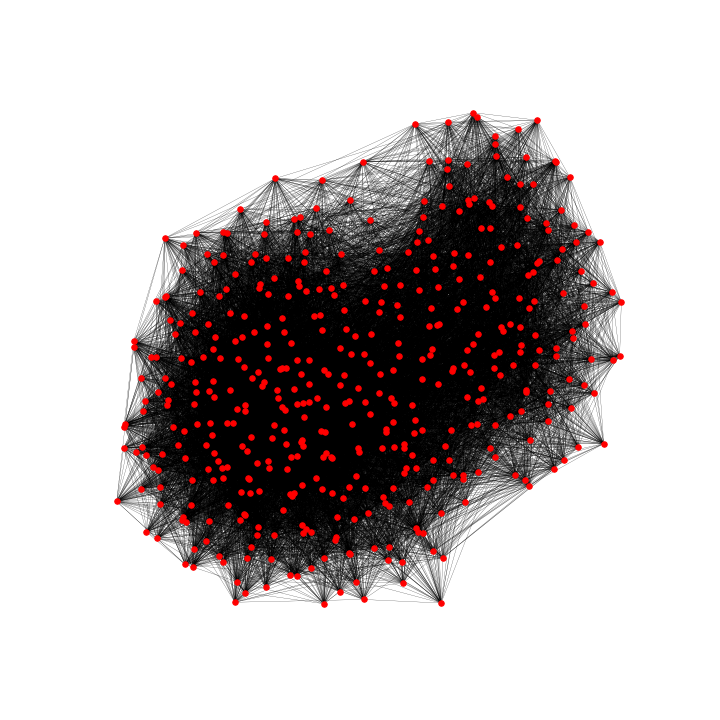} &
   \includegraphics[width=0.37\columnwidth]{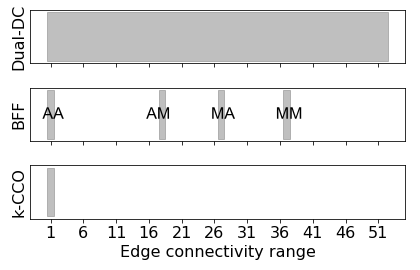} \\
   (a) &
     (b) & (c)  \\
 \end{tabular}
 \caption{\label{fig:introfig} Visualization of our algorithm's output on the Twitter retweet ($G$) and reply ($H$) graphs, with connectivity requirement $k=8$. (a) 8-edge-connected retweet subgraph.  (b) Dense reply subgraph with min degree 51. (c) The four variants (AA, AM, MA, and MM) of Best Friends Forever (BFF) \cite{semertzidis2019finding} result in four different increasing connectivity values, and \kcco~\cite{Cui2018} yields a simply connected subgraph on the dual Twitter graph (reply, quote). Only our \dc method allows to {\em fully control} the connectivity constraint. }
 \vspace{-5mm}
 \end{figure}

\noindent We shall also refer to $k$ value as the connectivity requirement.  Perhaps surprisingly, the choice of the minimum degree over the average degree as the density measure allows to solve Problem~\ref{prob:maxmin} efficiently. In contrast the average degree objective leads to an intractable formulation (see Section~\ref{sec:related}).   
 
\vspace{2mm}

$\bullet$ {\bf Algorithm design.} We prove that Problem~\ref{prob:maxmin}  is solvable in polynomial time. Our
proof is constructive, i.e., we design a polynomial-time exact algorithm that scales to large networks.
Figure~\ref{fig:introfig}(a) and (b) shows the output of our algorithm on two layers of the Twitter network, namely the retweet and the reply layers. Specifically, Figure~\ref{fig:introfig}(a) depicts the 8-edge-connected retweet subgraph, and  Figure~\ref{fig:introfig}(b) the resulting dense subgraph on the reply layer, both induced by the optimal set of nodes found by our algorithm. For more details,  see Section~\ref{sec:exp}.

$\bullet$ {\bf Ranging connectivity.} A key contribution of our work over prior methods is our ability to control the connectivity requirement. Figure~\ref{fig:introfig}(c) shows a preview of the range of connectivity values on the dual Twitter graph (reply, quote) that the existing approaches BFF~\cite{semertzidis2019finding} and \kcco~\cite{Cui2018} output. The plot for BFF is annotated by the specific version of the BFF problem (see Section~\ref{sec:related} for more details). Observe that maximizing the average degree over the union of the two graphs results in a densest subgraph that is 1-edge-connected. \kcco also ouputs a simply connected graph on $G$, that can be disconnected by the removal of a single edge.  While BFF-MM achieves high edge connectivity, only our method can {\em control} fully the connectivity requirement.  This is a powerful feature of our work, that enables the practitioner to obtain insights, not accessible by other competitors. 

\vspace{1.5mm}

$\bullet$  {\bf Application \#1: Mining multilayer networks.} We use our algorithmic primitive to mine different layers of the Twitter network. The characteristics of the two subgraphs induced by the optimal set of nodes provide insights into the different types of interactions of users on Twitter. 

\vspace{1.5mm}

$\bullet$ {\bf Application \#2: Mining human brain networks.} We use our algorithmic primitive on 101 human children brain datasets available from the Brain Imagine Data Exchange (ABIDE) project \cite{craddock2013neuro}.  Among these datasets, 52 correspond to typically developed (TD) children, and 49 to children with Autism Spectrum Disorder (ASD).  We show that our algorithm can be used to extract a clear signal of separation on average  between a pair of  brain networks corresponding to two typically developed (TD) kids, and a pair of  networks each corresponding to a TD child and a child suffering from ASD.

    \section{Related Work}
\label{sec:related}

\noindent \spara{Mining dual graphs.} All graphs considered in this paper are simple, unweighted, and undirected.  Wu et al. \cite{Wu+15,WuZLFJZ16} pose the question of finding a subset of nodes $S \subseteq V$ such that $G[S]$ forms a connected graph and $H[S]$ maximizes the average degree; they prove that this problem is \NPhard\ and designed a heuristic. Cui et al. introduce the $k$-connected core problem (k-CCO) for dual graphs. The k-CCO formulation takes as input a dual graph $(G,H)$ and a parameter $k$, and  aims to find a set of nodes $S \subseteq V$ that induced a connected $k$-core in $G$, and a simply connected subgraph in $H$. Semertzidis et al. introduce  the Best Friends Forever (BFF) problem  \cite{semertzidis2019finding}; see also Charikar et al.~\cite{charikar2018finding} for improved complexity and algorithmic results.  Specifically, Semertzidis et al. \cite{semertzidis2019finding} propose four different formulations for finding a subset of nodes $S \subseteq V$ that induces a dense subgraph across a collection of graphs with vertex set $V$. They  use the minimum and the average degree to measure edge density, and then maximize either the average or the minimum  measure of edge density across the collection of graphs. This results in four variants abbreviated as  MM, MA, AM, and AA.  The second letter indicates the density measure, while the first whether we consider the average (A) or the minimum (M) across the collection.  For example, AM aims to   maximize  the average  induced minimum degree over all graphs, whereas AA aims to maximize the average induced average degree over all graphs.  It is worth emphasizing that BFF variants impose no connectivity constraint, but according to the classic Mader's theorem we know that every graph of average degree (at least) $4k$ has a $k$-connected subgraph \cite{mader1972existenzn}. Table~\ref{tab:synopsis} summarizes the papers that lie closest to our work. 

Despite the existence of numerous heuristics for mining dual graphs, and more generally multi-layer networks, there are significantly fewer results related to optimizing concrete mathematical objectives simultaneously over two or more graphs  on the same set of nodes. Notably, Bhangale et al. recently designed a near-optimal algorithm for the simultaneous Max-Cut   problem , and proved a hardness result that shows that simultaneous optimization of Max-Cut over more than one graphs is harder than a single graph in terms of approximation \cite{bhangale2018near,bhangale2020simultaneous}, see also   \cite{bollobas2002problems,kuhn2005maximizing} for earlier results on judicious partitions.   There is a large amount of research work related to multi-layer graphs for other problems such as core decomposition \cite{Galimberti+17} and community detection  \cite{kim2015community}.   Yang et al. \cite{yang2018mining}, Tsourakakis et al. \cite{tsourakakis2019novel} study the problem of finding a subset of nodes that induces a dense subgraph on $G$ and a sparse subgraph on $H$. Also related to our work from an experimental point of view is the  work of Lanciano et al.  \cite{lanciano2020explainable}  who use  human brain networks to generate easy-to-interpret graph features that can  be used to diagnose autism disorders.

 \begin{table*}[]
\centering
\begin{tabular}{c|c|c|c}
{Methods} & Density measure  ($\max $) & \makecell{Connectivity\\ constraint on $G[S]$} & Hardness   \\
\hline
{Wu et al.~\cite{Wu+15}}      &  Avg. degree  $d_H(S)$       &  Connected       &          NP-hard  \\
\hline
{Cui et al.~\cite{Cui2018}}    &   Maximal $k$-core  \hide{ $|S|:\min \deg (H_S)\geq k$ }     &    Connected       &     P      \\
\hline
BFF-MM~\cite{semertzidis2019finding}          & $\min (\delta_G(S), \delta_H(S))$  &           N/A      & P~\cite{semertzidis2019finding} \\
BFF-AM          &    $ \delta_G(S)+\delta_H(S)$ &   N/A     &        NP-hard~\cite{charikar2018finding}    \\
BFF-MA          & $\min (d_G(S), d_H(S))$ & N/A     & NP-hard~\cite{charikar2018finding}\\
BFF-AA          & $d_G(S)+d_H(S)$   &         N/A              &  P \cite{Charikar00,semertzidis2019finding} \\
\hline
{\dc (Ours)}    &  $\delta_H(S)$    & $k$-edge connected     & P  \\
\end{tabular}
\caption{  \label{tab:synopsis} Comparison of our proposed framework \dc to other prior work. Here, for a graph $G$ and a subset of nodes $S\subseteq V_G$, quantities $d_G(S),\delta_G(S)$ are the average degree of $G[S]$ and the minimum induced degree $\text{min-deg}_G(S):=\min_{v \in S} \text{deg}_{G[S]}(v)$ respectively.}
\end{table*}

\spara{Dense subgraph discovery.} One of the most popular optimization models in DSD is the densest subgraph problem (DSP), which asks to find a subset of nodes $S\subseteq V$ that maximizes the average degree of $G[S]$. 
Unlike most of the other models, this problem is known to be polynomial-time solvable~\cite{Goldberg84}.  Furthermore, there exists a linear time, greedy $\frac{1}{2}$-approximation algorithm \cite{Asahiro+00,Charikar00} that was recently generalized and analyzed by Boob et al. \cite{boob2020flowless}  and  Chekuri, Quandru, and Torres \cite{chandra} respectively.
Furthermore, there exists a large number of DSP variations. The most well-studied variants are the size-restricted ones~\cite{Andersen_Chellapilla_09,Bhaskara+10,Feige+01,Khuller_Saha_09}.
For example, in the densest $k$-subgraph problem~\cite{Feige+01}, given a graph $G$ and a positive integer $k$,
we are asked to find $S\subseteq V$ that maximizes the average degree of $G[S]$ subject to the size constraint $|S|=k$.
It is known that such a restriction makes the problem much harder to solve;
in fact, the densest $k$-subgraph problem is \NPhard\ and the best known approximation ratio is $O(|V|^{{1/4}+\epsilon})$ for any $\epsilon>0$~\cite{Bhaskara+10}.
The average degree has recently been generalized in various ways 
to obtain more sophisticated structures~\cite{Miyauchi_Kakimura_18,Veldt,Kawase_Miyauchi_18,Mitzenmacher+15,Tsourakakis_15}. 
The densest subgraph problem has also been studied in various other settings, including the dynamic settings\cite{Bhattacharya+15,Epasto+15,Hu+17,sawlani2020near}, and the MapReduce model~\cite{Bahmani+12,bahmani2014efficient}. Finally, Bonchi et al. \cite{bonchi2021finding} very recently proposed a family of algorithms for finding the densest $k$-connected subgraphs on the single network topology; our current work extends this work to the dual graph setting.

\spara{Partitioning a graph into well-connected components.} A line of research closely related to DSD aims to partition a graph into well-connected components.  A major connectivity notion is \emph{$k$-edge-connectivity}. An unweighted graph $G$ is said to be $k$-edge-connected if the removal of any $k-1$ or fewer edges leaves $G$ connected, or equivalently, by Menger's theorem (see~\cite{diestel2000graphentheory}), if there are at least $k$ edge-disjoint paths between every pair of distinct vertices. Zhou et al.~\cite{Zhou+12} recently studied the problem to find a family of maximal $k$-edge-connected subgraphs.
To find that, a naive approach is to iteratively compute a minimum cut on each of connected components
until every connected component is either $k$-edge-connected or a singleton.
However, such an approach is computationally quite expensive and thus prohibitive for large graphs.
To overcome this issue, Zhou et al.~\cite{Zhou+12} devised some techniques and incorporated them into the naive approach.
Later, Akiba et al.~\cite{Akiba+13} and Chang et al.~\cite{Chang+13} developed much more efficient algorithms.
The algorithm by Akiba et al.~\cite{Akiba+13} runs in $O(|E|\log |V|)$ time,
whereas Chang et al.'s algorithm~\cite{Chang+13} runs in $O(hl|E|)$ time,
where $h$ is the height of the so-called decomposition tree and $l$ is the number of iterations of some subroutine,
both of which are instance-dependent parameters typically bounded by a small constant.

There are many other graph partitioning algorithms related to the above.
Hartuv and Shamir~\cite{Hartuv_Shamir_00} designed an algorithm to partition a graph into \emph{highly connected} subgraphs
(i.e., subgraphs with edge connectivity greater than half the number of its vertices).
Toida~\cite{Toida85} gave an efficient algorithm to find all maximal subgraphs with a pre-specified spanning-tree packing number.
The spanning-tree packing number of a graph $G$ is the maximum number of edge-disjoint spanning trees
that one can pack in $G$, which is closely related to the average degree of the densest subgraph  of $G$ (see, e.g.,~\cite{gusfield1983connectivity, catlin2009edge}).

    \section{Proposed Method}
\label{sec:proposed}

 We design a polynomial-time exact algorithm for Problem~\ref{prob:maxmin}. Our algorithm is shown in pseudocode as Algorithm~\ref{algo:maxmin}, which takes as input the two graphs $G,H$ on the same vertex set
$V$ and a parameter $k$ that specifies the $k$-edge connectivity constraint on $G$. The algorithm is recursive, and carefully breaks down $G$ in its $k$-edge connected components. Once the vertex set $S$ induces a $k$-edge connected component in $G$, the algorithm removes a node $v \in S$ having the lowest degree in $H$ and returns as its output the best between $S$ and
the result of a recursive invocation of the algorithm on $S\setminus \{v\}$.
The idea is to maintain the invariant that any solution which is $k$-connected in $G$ and with higher minimum $H$-degree than the
current best must be entirely contained in one of the subgraphs that we recurse into. Our main theoretical result concerns the correctness of our algorithm, and is stated as the next theorem. 

\begin{algorithm2e}[t]
\caption{\label{algo:maxmin} $\dc(G(V,E_G),H(V,E_H),k)$ }
\SetKwInput{Input}{Input}
\SetKwInput{Output}{Output}
\Input{\ $G=(V,E_G)$, $H=(V,E_H)$, $k\ge 1$ (edge connectivity)}  
\Output{\   $S^*$ such that the minimum degree in $H[S^*]$ is maximized subject to $G[S^*]$ being $k$-edge connected; or NO SOLUTION}  
\If{$G$ is not $k$-edge-connected}{ 
Let $\mathcal{F} \leftarrow \{C_1, \ldots, C_r\}$ be the maximal $k$-edge-connected components of $G$\; 
\If{$r = 0$}{
    \Return NO SOLUTION
}\Else{
\For{$i=1$ to $r$}{
    $S_i \leftarrow \dc(G[C_i],H[C_i],k)$\;
}
$t \leftarrow \argmax \{ \textrm{min-deg}_H(S_i) \mid i \in [r] \wedge S_i \neq \text{NO SOLUTION} \}$\;
\Return $S_t$
}
    
}
\ElseIf{$|V|>1$}{
Let $v$ be the vertex in $V$ 
    of minimum degree in $H[V]$ (ties broken arbitrarily)\; 
$T \leftarrow \dc(G[V\setminus \{v\}],H[V\setminus \{v\}],k)$\;
\If{$T$ = NO SOLUTION or $\textrm{min-deg}_H(T) \le \textrm{min-deg}_H(V)$ }{
    \Return{V}
}\Else{
    \Return{T}
}

}
\Else{
\Return{$V$}
}
\end{algorithm2e}

\begin{theorem}
\label{thrm1} 
Algorithm~\ref{algo:maxmin} outputs an optimal solution for Problem~\ref{prob:maxmin}. 
\end{theorem}
\begin{proof}
We argue by induction on the size of  the common vertex set $V$. Let $S^*$ denote the optimal solution for $(G, H)$. If $|V| \le 1$, then $S^* = V$ (Line 18).
So assume $|V| \ge 2$.  Consider first the case where $G$ is not $k$-connected.
Any $k$-connected subgraph of $G$ is entirely contained in one of the maximal
$k$-connected components $C_1, \ldots, C_r$, so one of them contains $S^*$, say $C_j$.
By induction, the
optimal solution within $C_j$ is computed in Line 7 (when $i = j$), and no other solution computed in Line 7 is feasible and has a higher minimum degree in $H$ than $S^*$ (otherwise
        $S^*$ would not be optimal).
Thus the algorithm returns an optimal solution with the same minimum degree as $S^*$ in Line 9. If $G$ is $k$-connected, let $v \in V$ be a vertex with the minimum degree in $H$.
We distinguish the following two cases:
\begin{description}
    \item[Case (i):]  $V$ is an optimal solution (i.e., $\textrm{min-deg}_H(V) = \textrm{min-deg}_H(S^*)$). Then $V$ is returned in Line 14. 
    \item[Case (ii):] $V$ is not an optimal solution. In this case $\textrm{min-deg}_H(V) < \textrm{min-deg}_H(S^*)$ holds. 
    But this implies that no optimal solution includes $v$, because the degree
    of $v$ in $H[S^*]$
    cannot be larger than its degree in $H$. 
    Hence Line 12 computes an optimal solution for the pair $(G[V], H[V])$ by the induction hypothesis, which is returned in Line 16.
\end{description}
\end{proof}

\spara{Time complexity.} An efficient implementation of Algorithm~\ref{algo:maxmin} runs in time $O(n m \log n)$ in the  RAM model, where $n$ is the number of nodes and $m > 1$ is the maximum number of edges in $G$ and $H$.
     To see this, observe that in Line 11 of Algorithm~\ref{algo:maxmin}, after removing the minimum degree vertex of degree $d$, we can iteratively remove all nodes of degree $d$ in the induced subgraph, i.e., find the set
of nodes $C$ in the $(d+1)$-core of $H$ and replace $G$ and $H$ with $G[C]$ and  $H[C]$ respectively.
This reduces the number of $k$-connected
component computations in the algorithm and doesn't affect its correctness because in Case (ii) above, none of the additional vertices removed can be part of an optimal solution: the minimum degree of an optimal solution must be at least $d+1$, so it must
be contained in the $(d+1)$-core of $H$.

Let $d$ be the minimum degree of the optimal solution in $H$ (or 0 if none exists). We argue inductively that the running time of Algorithm~\ref{algo:maxmin} (with the aforementioned faster implementation) is $O(n
        + (d + 1) m\log n) = O(n m \log n)$. 
Indeed, the total running time spent between Lines 11 and 16, excluding recursive calls, is $O(n + m)$, and the maximum depth of recursive calls is
at most $d$.
Each computation of $k$-edge-connected components in a subgraph with $n'\le n$ vertices and $m'\le m$ edges takes time $O(m' \log n')$ by~\cite{Akiba+13}. In Line 2, $r$ subgraphs are found with $m'_1,
    \ldots, m'_r$ and $\sum_i m'_i \le m'$, $\sum_i n'_i \le n'$. Each recursive invocation takes time $O(n'_i + (d+1) m'_i \log n)$ by induction, 
    for a total time of $O(\sum_i (n'_i + (d + 1) m'_i \log n)) = O(n + (d + 1) m \log n)$, as claimed. We summarize the above analysis with the following theorem statement.

\begin{theorem}
\label{thm:runtime} 
Algorithm~\ref{algo:maxmin} can be implemented to run in $O(nm\log n)$ time, where $n = |V|$ and $m = \max(|E_G|,|E_H|)$. 
\end{theorem}

\noindent \spara{Remarks.} Our algorithm naturally extends to the following version of Problem~\ref{prob:maxmin}, where we are given a graph $G$ and a collection of graphs $\mathcal{H}=\{H_1,\ldots,H_T\}$ and our goal is to find a set of nodes $S\subseteq V$ such that  $G[S]$ is $k$-edge connected, and the minimum degree across the collection $\mathcal{H}$ is maximized.  The only difference in the algorithm is that the peeling in Line 11 is done across the set of graphs $\mathcal{H}$, i.e., the node being removed is the node that has the smallest degree across $H_1,\ldots,H_T$. The proof follows our inductive proof and the argument in Proposition 1~\cite{semertzidis2019finding}. Finally, if instead of using the minimum degree $\delta_H(S)$, we use the average degree $d_H(S)$, the formulation of Problem~\ref{prob:prob1} becomes NP-hard.  This is a direct corollary of Wu et al.~\cite{Wu+15} for the special case of our problem with 1-edge connectivity constraint.

    \section{Experimental results}
\label{sec:exp}



\subsection{Setup}
\label{subsec:setup}
\hfill
 
\spara{Datasets.}   For our synthetic experiments, we generate graphs with stochastic block models~\cite{HOLLAND1983109}. The real-world Twitter datasets we use in our experiments are summarized in Table~\ref{tab:dualstats}.  We   experiment with Twitter multilayer
networks \cite{sotiropoulos2019twittermancer} crawled  from Twitter traffic generated during the month of February 2018 by
Greek-speaking users using a publicly available crawler {\sc twAwler} \cite{pratikakis2018twawler}, see also~\cite{sotiropoulos2019twittermancer}.  Specifically, we use four layers
each corresponding to  the type of interaction  {\it reply}, {\it quote}, {\it retweet}, and {\it follow}. For each pair of graphs we
test, we report the number of common nodes (i.e., Twitter accounts) that appear in both graphs, and the number  of edges in each graph.

We also use brain network datasets\cite{lanciano2020explainable} preprocessed from the public dataset released by the Autism Brain Imagine Data Exchange project. Experiments are done on 101  brain networks of children patients, 52 Typically Developed (TD) and 49 suffering from Autism Spectrum Disorder (ASD). Each network is undirected and unweighted with 116 nodes, summarizing patient's brain activity.

\spara{Machine specs.} The experiments were performed on a single machine, with Intel i7-10850H CPU @ 2.70GHz and 32GB of main memory.

\spara{Implementation and competitors.} Our code is available at \url{https://github.com/tsourakakis-lab/dense-kedge-connected}. We use the code of
Akiba et al. \cite{Akiba+13} to decompose our graph into $k$-edge connected components. Although we introduce Problem~\ref{prob:maxmin} for the first time, two competitors, i.e., BFF\footnote{We use the code from \url{https://github.com/ksemer/BestFriendsForever-BFF-}} and \kcco\footnote{The original code was not provided to us by the authors, so we provide our own implementation of \kcco in Python3.}, are considered in this section.

\begin{table}
	\caption{Statistics of dual graphs from Twitter multilayer networks and Enron (mail, cc) networks.}
	\label{tab:dualstats}
	\centering
	\begin{tabular}{c|c|c|c}
        \hline
		\# Dual graphs $(G, H)$ & \# common nodes  & \# edges in $G$ &  \# edges  in $H$ \\
		\hline
		(Reply, Quote) & 0.15M & 0.46M  & 0.48M \\
		(Reply, Retweet)  & 0.23M & 0.61M  & 2.14M \\
		(Retweet, Follow) & 0.32M & 2.39M  & 3.49M \\
		\hline
	\end{tabular}
\end{table}

\subsection{Ranging connectivity requirement}
\label{subsec:rangekconn}

An important aspect of our proposed framework is the fact that we can range the connectivity requirement in contrast to other methods that use different formulations or heuristics to mine dense subgraphs or communities from multilayer networks. We  illustrate the power of our framework  by comparing to the elegant Best-Friends-Forever (BFF) formulations  \cite{semertzidis2019finding}, as well as the \kcco model\cite{Cui2018}.  We generate two random graphs $G$ and $H$ on the same node set according to the stochastic block model.  Both graphs contain five  blocks $B_1, \ldots ,B_5$, where each block has 50 nodes. The internal edge density of each block $B_i$ in graph $G$, i.e., the probability any two nodes within $B_i$ are connected, is $0.1\cdot i$, $i=1,\ldots,5$; edges across blocks are generated with low probability $2\times 10^{-4}$ in order to ensure that the graph is connected, but not well-connected. The edge probability of block $B_i$ in graph $H$ is $0.1+0.1\cdot(5-i)$,  $i=1,\ldots,5$; edges across blocks are generated with probability $0.1$. Note the densities of blocks in $G$ are increasing from $B_1$ to $B_5$, while they are decreasing in $H$.

\begin{figure}[h]
    \begin{tabular}{cc}
		\includegraphics[width=0.48\textwidth]{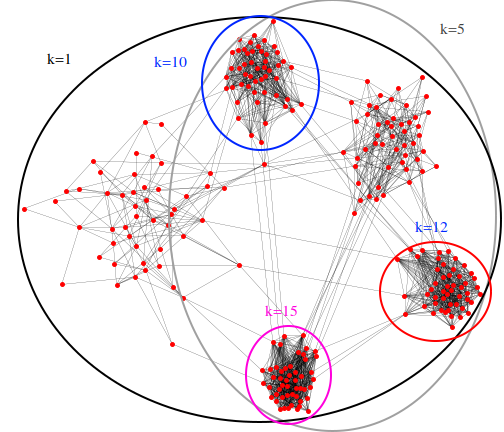} &    \includegraphics[width=0.48\textwidth]{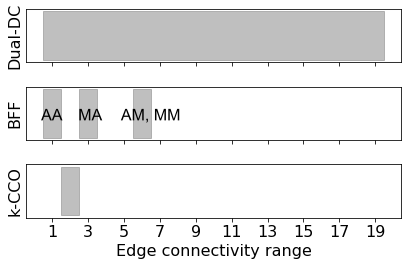} \\
		(a) &
		(b)    \\
	\end{tabular}
	\caption{\label{fig:synblock_krange} (a) Blocks found by our method for  different connectivity requirement $k$ values, visualized on graph $G$. (b) Subgraph connectivity on $G$ ranged by all methods.}
\end{figure}

\begin{figure}[h]
	\begin{tabular}{cccc}
		\includegraphics[width=0.23\textwidth]{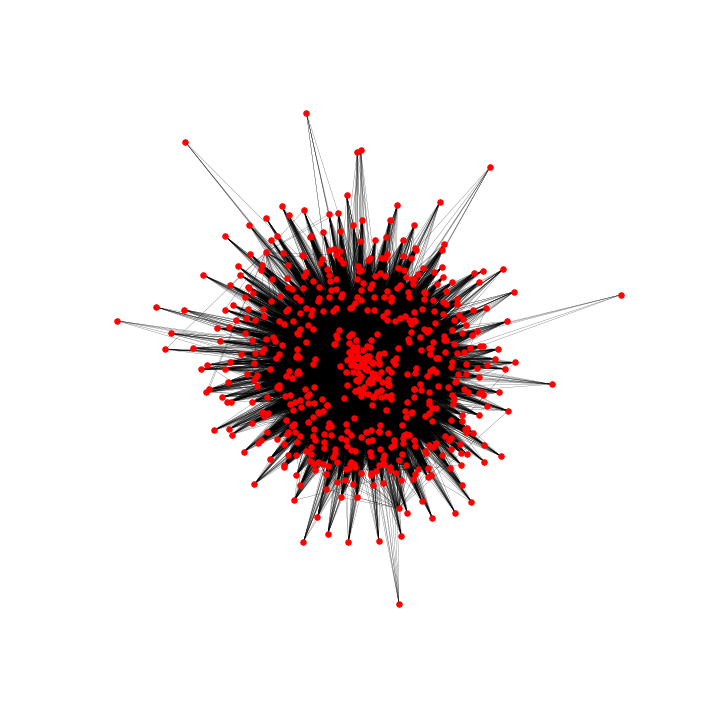} &    \includegraphics[width=0.23\textwidth]{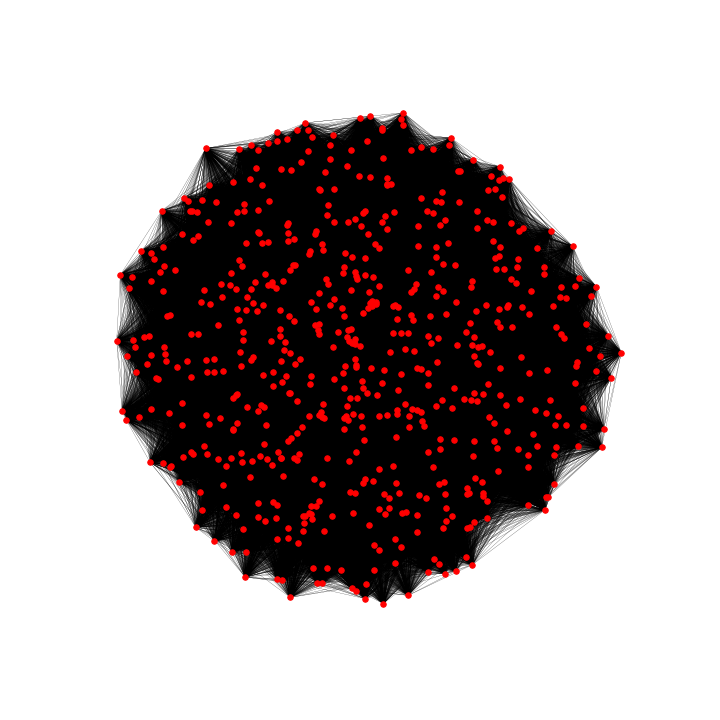} & \includegraphics[width=0.23\textwidth]{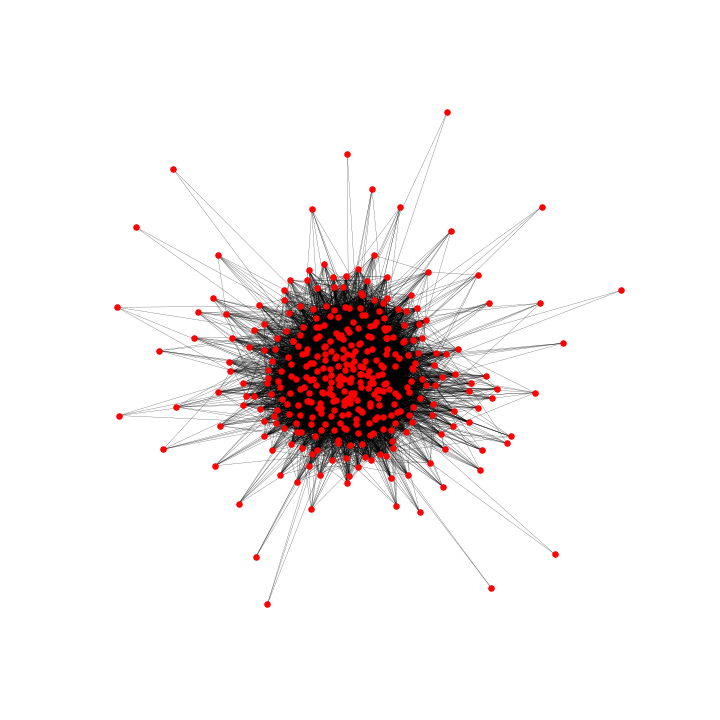} &    \includegraphics[width=0.23\textwidth]{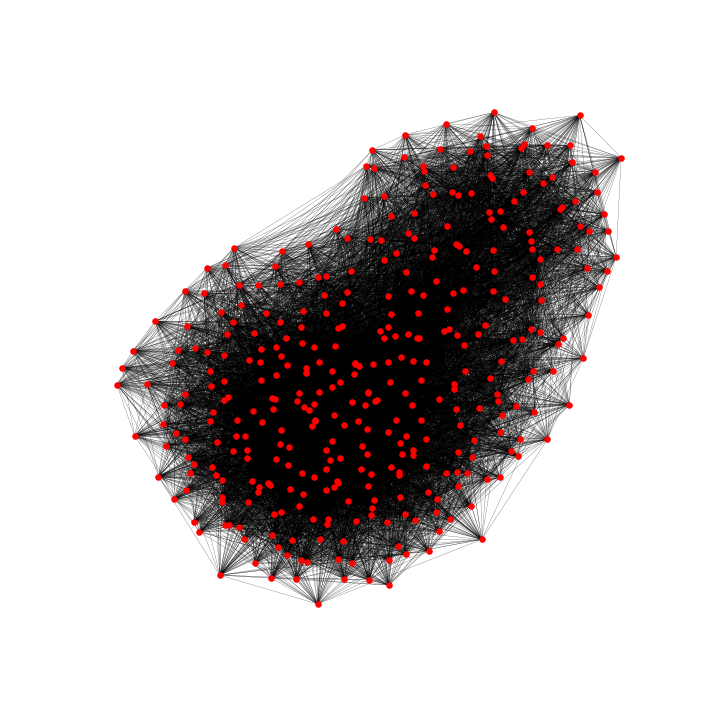}\\
		(a) &
		(b)   &
		(c) &
		(d)   \\
	\end{tabular}
	\caption{\label{fig:twitter} Visualization of our algorithm's output on the Twitter graphs, with connectivity requirement $k=2$. (a), (b) Follow$[S^*]$, Retweet$[S^*]$, (c),(d)  Quote$[S^*]$, Reply$[S^*]$.}
	\vspace{2mm}
\end{figure}

The BFF algorithm for the  AA formulation (see Section~\ref{sec:related} and \cite{semertzidis2019finding}) returns the whole graph, while the other three formulations (AM, MM, and MA)  return the subgraph induced by blocks $B_3\cup B_4\cup B_5$. The \kcco algorithm always returns the whole graph given different $k$ core value constraint on $H$. However, our method can mine the interesting connectivity structure for different $k$ values on $G$, and return the cluster has the highest core value on $H$. Figure~\ref{fig:synblock_krange}(a) visualizes the different blocks on $G$  obtained for different $k$ values, and Figure~\ref{fig:synblock_krange}(b) shows our method has more comprehensive $k$-edge connectivity control when comparing with benchmarks.

\begin{table}[hbt!]
\centering
	\caption{Twitter dual graph results. Statistics are calculated on subgraphs of $H$. }
	\label{tab:twitterresult}
	\scalebox{0.93}{
	\begin{tabular}{c|c|ccccccc}
	\hline
		Graph pair & $k$ & \makecell{\# of\\ nodes} & min deg & max deg & avg deg & diameter & \makecell{\# of\\ triangles} & \makecell{avg shortest\\ path} \\
		\hline
		(Reply, Quote) & 2 & 369 & 84 & 286 & 133.2 & 2 & 452\,293 & 1.64 \\
		(Reply, Quote) & 4 & 385 & 81 & 288 & 132.2 & 2 & 462\,063 & 1.66 \\
		(Reply, Quote) & 6 & 306 & 79 & 238 & 121.7 & 2 & 345\,408 & 1.6 \\
		(Reply, Quote) & 8 & 358 & 76 & 265 & 124.8 & 2 & 396\,485 & 1.65 \\
		(Reply, Quote) & 15 & 334 & 67 & 243 & 113.1 & 2 & 313\,266 & 1.66 \\
		(Reply, Quote) & 30 & 309 & 47 & 223 & 90.4 & 2 & 181\,080 & 1.71 \\
		(Reply, Quote) & 40 & 270 & 36 & 195 & 78.4 & 3 & 125\,975 & 1.71 \\
		\hline
		(Quote, Reply) & 2 & 368 & 51 & 222 & 81.8 & 3 & 148\,437 & 1.78 \\
		(Quote, Reply) & 4 & 419 & 50 & 242 & 82.9 & 3 & 167\,276 & 1.8 \\
		(Quote, Reply) & 6 & 334 & 50 & 203 & 78.8 & 3 & 130\,708 & 1.76 \\
		(Quote, Reply) & 8 & 465 & 49 & 257 & 83.9 & 3 & 185\,026 & 1.83 \\
		(Quote, Reply) & 15 & 300 & 47 & 189 & 73.3 & 3 & 105\,487 & 1.76 \\
		(Quote, Reply) & 30 & 294 & 41 & 181 & 67.6 & 3 & 89\,138 & 1.77 \\
		(Quote, Reply) & 40 & 276 & 36 & 172 & 62.1 & 3 & 71\,144 & 1.78 \\
		\hline
		(Reply, Retweet) & 2 & 470 & 275 & 468 & 362.9 & 2 & 8.4M & 1.23 \\
		(Reply, Retweet) & 4 & 543 & 266 & 537 & 381.5 & 2 & 10.2M & 1.29 \\
		(Reply, Retweet) & 6 & 501 & 256 & 494 & 360.8 & 2 & 8.6M & 1.28 \\
		(Reply, Retweet) & 8 & 496 & 247 & 489 & 352.7 & 2 & 8M & 1.29 \\
		(Reply, Retweet) & 15 & 432 & 224 & 426 & 316.6 & 2 & 5.7M & 1.26 \\
		(Reply, Retweet) & 30 & 277 & 152 & 274 & 213.7 & 2 & 1.7M & 1.23 \\
		(Reply, Retweet) & 40 & 663 & 99 & 465 & 197.8 & 3 & 2.8M & 1.71 \\
		\hline
		(Retweet, Reply) & 2 & 359 & 52 & 217 & 82.5 & 3 & 149\,782 & 1.77 \\
		(Retweet, Reply) & 4 & 359 & 52 & 217 & 82.5 & 3 & 149\,782 & 1.77 \\
		(Retweet, Reply) & 6 & 515 & 51 & 276 & 88 & 3 & 214\,812 & 1.84 \\
		(Retweet, Reply) & 8 & 462 & 51 & 261 & 86.1 & 3 & 192\,543 & 1.82 \\
		(Retweet, Reply) & 15 & 350 & 51 & 214 & 81.3 & 3 & 143\,712 & 1.77 \\
		(Retweet, Reply) & 30 & 771 & 49 & 321 & 88 & 3 & 297\,365 & 1.94 \\
		(Retweet, Reply) & 40 & 724 & 49 & 302 & 85.8 & 3 & 272\,389 & 1.94 \\
		\hline
		(Retweet, Follow) & 2 & 1030 & 219 & 966 & 356.6 & 2 & 9.1M & 1.65 \\
		(Retweet, Follow) & 6 & 931 & 198 & 874 & 324.9 & 2 & 6.9M & 1.65 \\
		(Retweet, Follow) & 15 & 791 & 180 & 749 & 293.4 & 2 & 5.0M & 1.63 \\
		(Retweet, Follow) & 30 & 638 & 154 & 603 & 251.6 & 2 & 3.2M & 1.60 \\
		(Retweet, Follow) & 50 & 451 & 134 & 435 & 209.0 & 2 & 1.7M & 1.54 \\
		\hline
		(Follow, Retweet) & 2 & 620 & 285 & 613 & 415.1 & 2 & 13.5M & 1.33 \\
		(Follow, Retweet) & 6 & 612 & 285 & 606 & 413 & 2 & 13.2M & 1.32 \\
		(Follow, Retweet) & 15 & 488 & 282 & 487 & 372.1 & 2 & 9.1M & 1.24 \\
		(Follow, Retweet) & 30 & 612 & 265 & 606 & 400.6 & 2 & 12.3M & 1.34 \\
		(Follow, Retweet) & 50 & 484 & 228 & 480 & 342.7 & 2 & 7.5M & 1.29 \\

		\hline
	\end{tabular}
	}
\end{table}

\subsection{Mining Twitter}
\label{subsec:twitter}

Table~\ref{tab:twitterresult} shows our results for some pairs of Twitter graphs, and for various values of the connectivity requirement $k$.  We show the output difference by using the pairs (Reply, Quote),  (Reply, Retweet), and (Retweet, Follow) and their reverse ordering. Recall that  for a given  pair $(G,H)$ we impose the connectivity requirement on $G$. Table~\ref{tab:twitterresult} shows the number of nodes in the optimal solution $S^*$, as well as some basic graph statistics of the subgraph $H[S^*]$. We observe that $|S^*|$  is largest for the {\em follow} and {\em retweet} pairs of interactions. If this were not the case, this would have been surprising since the  corresponding 2-layered graph for  {\em follow} and {\em retweet}  shares more nodes (0.32M) than the other two pairs of interactions.    The average shortest path among all induced subgraphs on $H$ is always less than 2, and as can be seen by comparing $|S^*|$ and the maximum degree, there is typically a hub node connecting almost all pairs via its ego-network. Furthermore, as we can see by the average degree in $H[S^*]$, those induced subgraphs are quasi-clique-like \cite{Tsourakakis+13}.  Figure~\ref{fig:twitter} visualizes the output for our algorithm for the pairs ({\em follow, retweet}), ({\em quote,reply}) when $k=2$.

The problem of finding large-near cliques in a single graph is NP-hard, but in recent years  tools that work efficiently on large-scale
networks have been proposed (see \cite{konar2020mining,Tsourakakis+13,Tsourakakis_15,Mitzenmacher+15}).  It is worth emphasizing an
interesting side-effect of our algorithm that appears to hold on real-world dual graphs. Once you are able to find a well-connected
subgraph across two networks, it appears in practice that it is a  large near-clique. Our findings on the large near-cliques we find on
$H$, agree with the theorems of Konar and Sidiropoulos~\cite{konar2020mining} concerning the existence of large-near cliques in the ego-networks of certain nodes.


\subsection{Mining brain networks}
\label{subsec:fmri}

Finally we apply our algorithm to all possible  dual graphs defined by typically developed (TD) children and children suffering from Autism Spectrum Disorder (ASD). We report our findings for the $52 \times 51$ possible (i.e., ordered) pairs of brain graphs of TD children, and $49 \times 52$ possible pairs of children suffering from   ASD and TD children. 

Figure~\ref{fig:brain_runtime_compare}(a) shows the box plot  of the output sizes 
for $k=14$. Our results are stable over the choice of $k$ in the range we tried (i.e., $k$ values between 10 and 20).   Despite the
existence of several outlier pairs of (TD,TD) graphs with respect to their output size (e.g., plenty of  (TD,TD)  pairs share  about 40
        nodes as the joint optimal solution), there  is still a separation of the averages; for $k=14$ the respective average value of
$|S^*|$ is 95 and 90 for  (TD,TD), (ASD,TD) dual graphs respectively. We observe that even if the range of values is similar for the
two types of dual graphs, the medians are also separated as shown by the box plot. We conclude that  there exists  a weak but
measurable signal  that indicates that  healthy individual brains are at least on average well connected across a larger subset of nodes, a finding that agrees with \cite{lanciano2020explainable} in spirit, and is consistent with other studies in the context of other diseases that argue that  ``better connected brains, healthier brains''; see, e.g., \cite{supekar2008network,navlakha2015decreasing} and references therein. We highlight the importance of ranging connectivity to observe such phenomenon, as it is invisible from the results of \kcco that always returns the whole graph.

\begin{figure}[h]
    \begin{tabular}{cc}
		\includegraphics[width=0.48\textwidth]{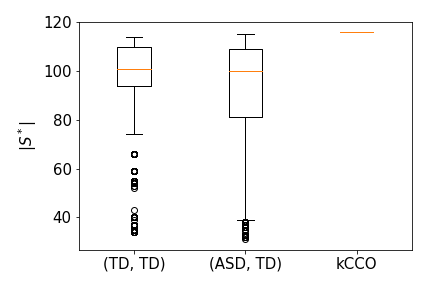} &    \includegraphics[width=0.48\textwidth]{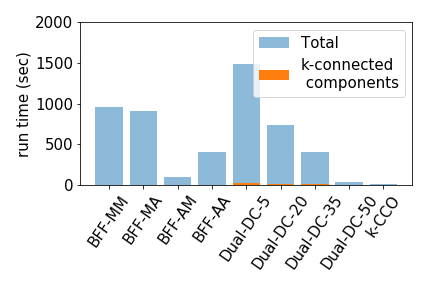} \\
		(a) &
		(b)  \\
	\end{tabular}
	\caption{\label{fig:brain_runtime_compare} (a) Box plots for the size $|S^*|$, over all possible (TD, TD) and (TD, ASD) dual graphs. The average size of $|S^*|$ is 95 and 90 respectively. The competitor \kcco always returns the whole graph. (b) Running time(sec) of Our algorithm with ranged $k$ value, together with four variants of BFF and \kcco.}
\end{figure}

\hide{

	\begin{figure}
		\begin{tabular}{cc}
			\includegraphics[width=0.43\columnwidth]{figs/brain/brain_k=14.png} & 	\includegraphics[width=0.43\columnwidth]{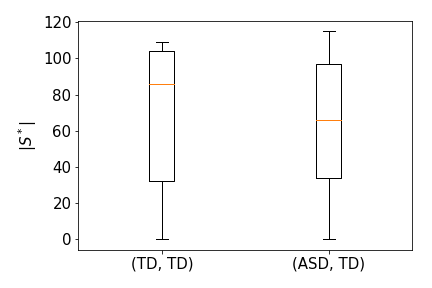}\\
			(a) & (b) \\
		\end{tabular}
		\caption{\label{fig:brain} Box plots}
	\end{figure}

	\begin{figure}[t]
		\centering
		\begin{subfigure}[]{\columnwidth}
			\centering
			\includegraphics[width=\columnwidth]{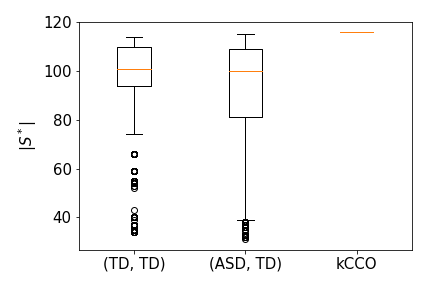}
			\caption{$k=14$, average sizes are 95 and 90.}
			\label{fig:braink=14}
		\end{subfigure}
		\hfill
		\begin{subfigure}[]{\columnwidth}
			\centering
			\includegraphics[width=\columnwidth]{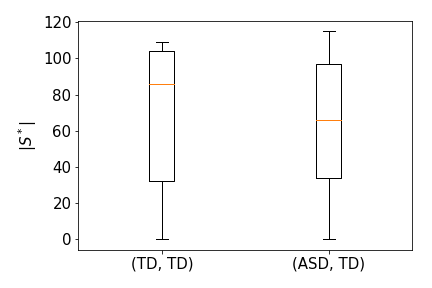}
			\caption{$k=15$, average sizes are 68 and 65.}
			\label{fig:braink=15}
		\end{subfigure}
		\caption{Box plots of $|S^*|$ of dual graphs (TD, TD) and (ASD, TD).}
		\label{fig:brain_box}
	\end{figure}
}

\hide{ 
\begin{figure}[h]
	\centering
	\includegraphics[width=\columnwidth]{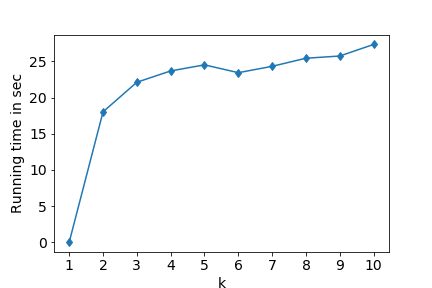}
		\captionof{figure}{Running time of Akiba et al. \cite{Akiba+13} algorithm versus the connectivity requirement $k$  on an Erd{\H{o}}s-R{\'e}nyi graph with $10\,000$ nodes and edge probability $0.005$.}
		\label{fig:synruntime_k}
\end{figure}
}

\subsection{Scalability analysis}
\label{subsec:synth}

Figure~\ref{fig:brain_runtime_compare}(b) shows the running time of the four variants of BFF, our method for $k=5,\, 20,\, 35,\, 50$,  and \kcco algorithms on the Twitter (reply, quote) dual graph respectively. We  observe that controlling the connectivity requirement comes at the cost of the run time, compared to the competitor methods.   \dc runs in at most 25 minutes for $k=5$. As $k$ grows, the depth of the recursion decreases,  and thus we observe that the total run time decreases and becomes comparable to the competitor methods. For instance, for $k=5,\, 50$ the depth of the recursion (i.e., the number of calls to \dc) is equal to 294 and 20 respectively. The bar plot illustrates the run time required for the $k$-connected components computation for this specific dataset.   In general, our method handles the graphs from Table~\ref{tab:dualstats}   in at most thirty minutes on a single machine for any $k$ connectivity value, with the single exception of the largest pair (Retweet, Follow), for which our code requires 5 hours to execute.
  
    \section{Conclusions}
\label{sec:conclusion}

In this work we introduce a new problem on a dual graph  $(G,H)$, that aims to find  a set of nodes that induces a well-connected subgraph on $G$ and a dense subgraph on $H$. We prove that our formulation admits an exact solution, and we propose an algorithm that runs in polynomial time.   In practice, our algorithm scales on graphs with several millions of edges on a single machine, and runs reasonably fast. Compared to competitor methods, our \dc method enables to control the connectivity constraint on $H$. Designing a faster algorithm is an interesting  open question.  We show that our method can be used in practice to mine layers of Twitter and human brain networks.  

Our work opens various interesting directions. A broad direction is the design of exact, and approximation algorithms for mining dual graphs. In principle, optimizing objectives over dual graphs is harder than over a single graph, not only from a formal complexity point of view (e.g., \cite{bhangale2020simultaneous}), but even for designing well-performing heuristics.  A natural open question is whether we can design an efficient approximation algorithm  when we choose the average degree instead of the minimum degree as density measure.

    \bibliographystyle{alpha}
    \bibliography{ref}
\end{document}